\documentclass[12pt]{elsarticle}
\usepackage[margin=2.5cm]{geometry}
\usepackage{lipsum}

\makeatletter
\def\ps@pprintTitle{%
   \let\@oddhead\@empty
   \let\@evenhead\@empty
   \def\@oddfoot{\reset@font\hfil\thepage\hfil}
   \let\@evenfoot\@oddfoot
}
\makeatother

\usepackage{amsfonts,color,morefloats,pslatex}
\usepackage{amssymb,amsthm, amsmath,latexsym}

\newtheorem{theorem}{Theorem}

\newtheorem{example}[theorem]{Example}

\newcommand{\tr}{{\mathrm{Tr}}}

\newcommand{\gf}{{\mathrm{GF}}}
\newcommand{\PG}{{\mathrm{PG}}}

\newcommand{\support}{{\mathrm{suppt}}}

\newcommand{\cA}{{\mathcal{A}}}

\newcommand{\C}{{\mathsf{C}}}
\newcommand{\E}{{\mathsf{E}}}

\newcommand{\cH}{{\mathcal{H}}}

\newcommand{\bc}{{\mathbf{c}}}
\newcommand{\bg}{{\mathbf{g}}}

\newcommand{\bzero}{{\mathbf{0}}}

\newcommand{\PGL}{{\mathrm{PGL}}}

\usepackage{blindtext}

\begin{document}

\begin{frontmatter}



\title{Near MDS codes from oval polynomials}

\tnotetext[fn1]{
Q. Wang's research was supported by the National Natural Science Foundation of China under grant number 61602342,
the Natural Science Foundation of Tianjin  under grant number 18JCQNJC70300, and the Science and Technology Development Fund of Tianjin Education Commission for Higher Education under grant number 2018KJ215.
Z. Heng's research was supported by the Natural Science Foundation of China under grant number 11901049,  the Natural Science Basic Research Program of Shaanxi (Program No. 2020JQ-343), and the Fundamental Research Funds for the Central Universities, CHD, under grant number 300102129301.}

\author[qywang]{Qiuyan Wang}
\ead{wangyan198801@163.com}
\author[zlheng]{Ziling Heng}
\ead{zilingheng@163.com}

\address[qywang]{School of Computer Science and Technology, Tiangong University,
 Tianjin 300387, China}
\address[zlheng]{School of Science, Chang'an University, Xi'an 710064, China}




\begin{abstract}
A linear code with parameters of the form $[n, k, n-k+1]$  is referred to as an MDS (maximum distance separable) code.
A linear code with parameters of the form $[n, k, n-k]$ is said to be almost MDS (i.e., almost maximum distance separable) or AMDS for short.
A code is said to be near  maximum distance separable (in short, near MDS or NMDS) if both the code and its dual
are almost maximum distance separable.  Near MDS codes correspond to interesting objects in finite geometry and have nice
applications in combinatorics and cryptography. In this paper, seven infinite families of
 $[2^m+1, 3, 2^m-2]$ near MDS codes over $\gf(2^m)$ and seven infinite families of  $[2^m+2, 3, 2^m-1]$ near MDS
 codes over $\gf(2^m)$ are constructed  with special oval polynomials for odd $m$. In addition, nine infinite families of
 optimal $[2^m+3, 3, 2^m]$ near MDS
 codes over $\gf(2^m)$ are constructed  with oval polynomials in general.
\end{abstract}


\end{frontmatter}


\section{Introduction}

Before introducing the motivations and objectives of this paper, we need to recall arcs in the projective plane
$\PG(2, 2^m)$ and oval polynomials over $\gf(2^m)$, and define near MDS codes.

\subsection{Almost MDS codes and near MDS codes}

A linear code with parameters of the form $[n, k, n-k+1]$ is called an MDS (maximum distance separable) code.
A linear code with parameters of the form  $[n, k, n-k]$ is said to be almost maximum distance separable (almost MDS or AMDS for short).
A code is said to be near  maximum distance separable (near MDS or NMDS for short) if both the code and its dual
are almost maximum distance separable.
By definition,
an $[n, k]$ linear code $\C$ over $\gf(q)$ is NMDS if and only if $d(\C) + d(\C^\perp)=n$, where $d(\C)$ and
$d(\C^\perp)$ denote the minimum distance of $\C$ and $\C^\perp$, respectively.
NMDS codes and $n$-tracks in finite geometry are closely related. The reader is referred to \cite{DeBoer96,DodLan95,DodLan00}
for further information of $n$-tracks in finite geometry and their connections with NMDS codes.

The existence of NMDS codes is of course a concern. It is known that
algebraic geometric $[n, k, n-k]$ NMDS codes over $\gf(q)$ for $q=p^m$ do exist for every $n$ with
\begin{eqnarray*}
n \leq
\left\{
\begin{array}{ll}
q + \lceil 2 \sqrt{q} \rceil & \mbox{ if $p$ divides $\lceil 2 \sqrt{q} \rceil$ and $m$ is odd,} \\
q + \lceil 2 \sqrt{q} \rceil +1 & \mbox{ otherwise,}
\end{array}
\right.
\end{eqnarray*}
and arbitrary $k \in \{2,3, \ldots, n-2\}$ \cite{TVlad}.

The first near MDS code was the $[11, 6, 5]$ ternary Golay code discovered in 1949 by Golay \cite{Golay49},
which has applications in group theory and combinatorics.
Some recent progress
in near MDS codes were made in \cite{DingTang19,KJ19,TangDing20,TD13}.

\subsection{Hyperovals, oval polynomials and $[q+2, 3, q]$ MDS codes over $\gf(q)$}\label{sec-3objects}

From now on let $q=2^m$ for a positive integer $m$. It is known that the automorphism group of the Desarguesian
projective plane $\PG(2, q)$ is the projective general linear group $\PGL_3(q)$.
An \emph{arc} in the Desarguesian projective plane $\PG(2, q)$
is a set of at least three points in $\PG(2, q)$ such that no three of them are collinear, i.e., no three of them are on the same line.
For any arc $\cA$ of the Desarguesian projective plane $\PG(2, q)$, it is known that $|\cA| \leq q+2$.
A \emph{hyperoval} $\cH$ in $\PG(2,q)$ is a set of $q+2$ points such that no three of
them are collinear, i.e., an arc in $\PG(2,q)$ with $q+2$ points. It is known that any line in
 $\PG(2,q)$ intersects with a hyperoval in $\PG(2,q)$ in either zero or two points.
Hyperovals are maximal arcs, as they have the maximal number of points as arcs.
Two hyperovals
are said to be \emph{equivalent} if there is
an automorphism of $\PG(2, q)$ that sends one to the other.

The theorem below shows that all hyperovals in $\PG(2,q)$ can be constructed with
a special type of permutation polynomials of the finite field $\gf(q)$ \cite[p. 504]{LN97}. It was
discovered by Segre.

\begin{theorem}\label{thm-hyperovaloply}
Let $m \geq 2$. Any hyperoval in the Desarguesian projective plane $\PG(2, q)$ can be written in the
following form
$$
\cH(f)=\{(f(c), c, 1): c \in \gf(q)\} \cup \{(1,0,0)\}  \cup \{(0,1,0)\},
$$
where $f \in \gf(q)[x]$ is a polynomial such that
\begin{enumerate}
\item $f$ is a permutation polynomial of $\gf(q)$ with $\deg(f)<q$ and $f(0)=0$, $f(1)=1$;  and
\item for each $a \in \gf(q)$, $g_a(x):=(f(x+a)+f(a))x^{q-2}$ is also a permutation polynomial
      of $\gf(q)$.
\end{enumerate}
Conversely, every such set $\cH(f)$ is a hyperoval.
\end{theorem}

Any polynomial satisfying the two conditions of Theorem \ref{thm-hyperovaloply} is
called an \emph{oval polynomial} (in short, an o-polynomial). For example,
$f(x)=x^2$ is an oval polynomial over $\gf(q)$ for all $m \geq 2$.
Two oval polynomials are said to be \emph{equivalent} if their hyperovals are equivalent.

Hyperovals in $\PG(2, q)$ and MDS codes over $\gf(q)$ with parameters $[q+2, 3, q]$
are equivalent objects in the sense that they can be constructed from each other. Below
we introduce their equivalence.

Given a hyperoval $\cH=\{h_1, h_2, \ldots, h_{q+2}\}$ in $\PG(2, q)$, one
constructs a linear code $\C_{\cH}$ of length $q+2$ over $\gf(q)$ with generator matrix
$[h_1, h_2, \ldots, h_{q+2}]$, where each $h_i$ is a column vector of the vector space $\gf(q)^3$.
It was shown in \cite[Section 12.2]{Dingbook18} that $\C_{\cH}$ is an MDS code over $\gf(q)$ with parameters $[q+2, 3, q]$
and weight enumerator
$$
1 + \frac{(q+2)(q^2-1)}{2} z^q + \frac{q(q-1)^2}{2} z^{q+2}.
$$
The dual of $\C_{\cH}$ is clearly an MDS code over $\gf(q)$ with parameters $[q+2, q-1, 4]$.

Conversely, given an MDS code $\C$ over $\gf(q)$ with parameters $[q+2, 3, q]$, one constructs
a hyperoval in $\PG(2, q)$ as follows. Let
$[h_1, h_2, \ldots, h_{q+2}]$ be a generator matrix of $\C$. Let $a_i \in \gf(q)^*$
such that $\bar{h}_i=a_ih_i$ is a point of $\PG(2, q)$. Then
$$
\cH=\{\bar{h}_1, \bar{h}_2, \ldots, \bar{h}_{q+2}\}
$$
is a hyperoval in the Desarguesian projective plane $\PG(2, q)$.

\subsection{Motivations and objectives of this paper}

The discussion in Section \ref{sec-3objects} showed that the following are equivalent objects:
\begin{itemize}
\item Oval polynomials over $\gf(q)$.
\item Hyperovals in $\PG(2, q)$.
\item $[q+2, 3, q]$ MDS codes over $\gf(q)$.
\end{itemize}
Hence, every oval polynomial over $\gf(q)$ gives a $[q+2, 3, q]$ MDS codes over $\gf(q)$. A natural question is whether oval polynomials over $\gf(q)$ can be used to construct near MDS codes. This paper is mainly motivated by this
question.

MDS codes are widely used in communication and data storage systems. However, the support designs of MDS codes
are complete and thus trivial \cite{Dingbook18}. Near MDS codes are not optimal with respect to the Singleton bound,
but may give nice $t$-designs \cite{DingTang19,TangDing20}. Hence, near MDS codes could be more interesting than
MDS codes in the theory of combinatorial designs. In fact, two 70-year breakthroughs were recently made by near MDS
codes in \cite{DingTang19,TangDing20}.  This is the first motivation of this paper.
The second and third motivations of studying near MDS codes are their applications in the design of block ciphers \cite{LW17}
and secret sharing \cite{Zhou09}.

In this paper, we construct seven infinite families of $[2^m+1, 3, 2^m-2]$ near MDS codes over $\gf(2^m)$ and seven infinite families of  $[2^m+2, 3, 2^m-1]$ near MDS codes over $\gf(2^m)$  for odd $m$ with special oval polynomials. We also present nine  infinite families of
 $[2^m+3, 3, 2^m]$ near MDS codes over $\gf(2^m)$, which are distance-optimal. We will determine the parameters
of the binary subfield codes of some of these near NMDS codes.

\section{Preliminaries}

\subsection{Some properties of  NMDS codes}

In this subsection, we introduce two basic results about NMDS codes that will be needed in this paper later.
We have the following weight distribution formulas for NMDS codes.

\begin{theorem}[\cite{DodLan95}]\label{thm-DLwtd}
Let $\C$ be an $[n, k, n-k]$ near MDS code over the finite field $\gf(q)$. Then the weight enumerators of the two codes
$\C^\perp$ and $\C$
are given by
\begin{eqnarray*}
A_{k+s}^\perp = \binom{n}{k+s} \sum_{j=0}^{s-1} (-1)^j \binom{k+s}{j}(q^{s-j}-1) +
             (-1)^s \binom{n-k}{s}A_{k}^\perp
\end{eqnarray*}
for $s \in \{1,2, \ldots, n-k\}$; and
\begin{eqnarray*}
A_{n-k+s} = \binom{n}{k-s} \sum_{j=0}^{s-1} (-1)^j \binom{n-k+s}{j}(q^{s-j}-1) +
             (-1)^s \binom{k}{s}A_{n-k}
\end{eqnarray*}
for $s \in \{1,2, \ldots, k\}$.
\end{theorem}

It was pointed out in \cite{DingTang19}
 that two $[n, k, n-k]$ NMDS codes over $\gf(q)$ could have different weight distributions.
This means that
the weight distribution of an $[n, k, n-k]$ NMDS code over $\gf(q)$ depends on not only  $n$,
$k$ and $q$, but also some other parameters of the code \cite{DingTang19}.
However,  the weight distribution of any $[n, k, n-k+1]$ MDS code over $\gf(q)$
depends only on $n, k$ and $q$.
This is a big difference between MDS
codes and NMDS codes.
The following theorem describes a nice property of NMDS codes and will be needed in the sequel when we
settle the weight distributions of some families of near MDS codes.

\begin{theorem}[\cite{FaldumWillems97}]\label{thm-121FW}
Let $\C$ be an NMDS code. Then for every minimum weight codeword $\bc$ in $\C$, there exists,
up to a multiple, a unique minimum weight codeword $\bc^\perp$ in $\C^\perp$ such that
$\support(\bc) \cap \support(\bc^\perp)=\emptyset$, where $\support(\bc)=\{1 \leq i \leq n: c_i \neq 0\}$
denotes the support of the codeword $\bc=(c_1, \ldots, c_n)$. In particular, $\C$ and $\C^\perp$
have the same number of minimum weight codewords.
\end{theorem}

The theorem above shows that there is a natural
correspondence between the minimum weight codewords of an NMDS code $\C$ and those of its dual
$\C^\perp$.

\subsection{Oval polynomials and their properties}

To construct near MDS codes over $\gf(q)$ in the sequel, we need specific oval polynomials over $\gf(q)$ and have to introduce
some of their properties. The following is a list of known infinite families of oval polynomials in the literature.

\begin{theorem}\label{thm-knownopolys}
Let $m \geq 2$ be an integer. The following are oval polynomials of $\gf(q)$, where $q=2^m$.
\begin{itemize}
\item The translation polynomial $f(x)=x^{2^h}$, where $\gcd(h, m)=1$.
\item The Segre polynomial $f(x)=x^6$, where $m$ is odd.
\item The Glynn oval polynomial $f(x)=x^{3 \times 2^{(m+1)/2} +4}$, where $m$ is odd.
\item The Glynn oval polynomial $f(x)=x^{ 2^{(m+1)/2} + 2^{(m+1)/4} }$ for $m \equiv 3 \pmod{4}$.
\item The Glynn oval polynomial $f(x)=x^{ 2^{(m+1)/2} + 2^{(3m+1)/4} }$ for $m \equiv 1 \pmod{4}$.
\item The Cherowitzo oval polynomial $f(x)=x^{2^e}+x^{2^e+2}+x^{3 \times 2^e+4},$ where $e=(m+1)/2$ and $m$ is odd.
\item The Payne oval polynomial $f(x)=x^{\frac{2^{m-1}+2}{3}} + x^{2^{m-1}} + x^{\frac{3 \times 2^{m-1}-2}{3}}$,
        where $m$ is odd.
\item The Subiaco polynomial
$$
f_a(x)=((a^2(x^4+x)+a^2(1+a+a^2)(x^3+x^2)) (x^4 + a^2 x^2+1)^{2^m-2}+x^{2^{m-1}},
$$
where $\tr_{q/2}(1/a)=1$ and $a \not\in \gf(4)$ if $m \equiv 2 \bmod{4}$.
\item The Adelaide oval polynomial
$$
f(x)=\frac{T(\beta^m)(x+1)}{T(\beta)} + \frac{T((\beta x + \beta^q)^m)}{T(\beta) (x+T(\beta)x^{2^{m-1}} +1)^{m-1}} + x^{2^{m-1}},
$$
where $m \geq 4$ is even, $\beta \in \gf(q^2) \setminus \{1\}$ with $\beta^{q+1}=1$, $m \equiv \pm (q-1)/3 \pmod{q+1}$,
and $T(x)=x+x^q$.
\end{itemize}
\end{theorem}

The following property of oval polynomials will be needed later.

\begin{theorem}[\cite{Masch98}]\label{thm-opoly2to1}
A polynomial $f$ over $\gf(q)$ with $f(0)=0$ is an oval polynomial if and only if $f_u:=f(x)+ux$
is $2$-to-$1$ for every $u \in \gf(q)^*$.
\end{theorem}

The next theorem gives another characterisation of oval polynomials, where the conditions are called
the slope condition, and will be needed later.

\begin{theorem}\label{thm-J22220}
$f$ is an oval polynomial over $\gf(q)$ if and only if
\begin{enumerate}
\item $f$ is a permutation of $\gf(q)$; and
\item
$$
\frac{f(x)+f(y)}{x+y} \neq \frac{f(x)+f(z)}{x+z}
$$
for all pairwise-distinct $x, y, z$ in $\gf(q)$.
\end{enumerate}
\end{theorem}

\begin{theorem}\label{thm-J2220}
Let $m \geq 3$ be odd and let $f(x)$ be an oval polynomial over $\gf(q)$ with coefficients in $\gf(2)$. Then $f(x)+x+1=0$ has no solution in $\gf(q)$.
\end{theorem}

\begin{proof}
By definition, $0$ and $1$ are not solutions of $f(x)+x+1=0$. Suppose $x \in \gf(q) \setminus \{0,1\}$ is a solution of
$f(x)+x+1=0$. Then $x^{2^h}$ is a solution of  $f(x)+x+1=0$ for each nonnegative integer $h$, as the coefficients of
$f(x)$ are in $\gf(2)$ by assumption. In particular, $x$, $x^2$
and $x^4$ are solutions of  $f(x)+x+1=0$.  By Theorem \ref{thm-opoly2to1}, the equation $f(x)+x+1=0$ has at most
two solutions. Since $x \not\in \{0,1\}$, $x^2 \neq x$ and $x^4 \neq x^2$. It then follows that $x^4=x$. Consequently,
$x^3=1$. Since $m$ is odd, $\gcd(2^m-1, 3)=1$. It then follows from $x^3=1$ that $x=1$, which is contrary to the
assumption that $x \not\in \{0,1\}$. This completes the proof.
\end{proof}

Let $m$ be even, and let $\alpha$ be a generator of $\gf(q)^*$. It is easily seen that $\alpha^{(2^m-1)/3}$ is a solution
of $x^2+x+1=0$. Hence, Theorem \ref{thm-J2220} is not true for even $m$.

Theorem \ref{thm-J2220} looks simple, but will play an important role in constructing near MDS codes in this paper.

\section{Near MDS codes with parameters $[q+3, 3, q]$ from oval polynomials}\label{sec-J271}

Let $f$ be a polynomial over $\gf(q)$ with $f(0)=0$ and $f(1)=1$. Let $\alpha$ be a generator of $\gf(q)^*$.
Define
\begin{eqnarray*}
B_f=\left[
\begin{array}{lllllll}
f(0) & f(\alpha^0)  & f(\alpha^1) & \cdots & f(\alpha^{q-2}) & 1 & 0 \\
0 &\alpha^0  & \alpha^1 & \cdots & \alpha^{q-2} & 0 & 1 \\
1      & 1         &   1           & \cdots & 1                    & 0 & 0
\end{array}
\right].
\end{eqnarray*}
By definition, $B_f$ is a $3$ by $q+2$ matrix over $\gf(q)$. Let $\E_f$ denote the linear code over $\gf(q)$ with
generator matrix $B_f$.  As informed in Section \ref{sec-3objects}, $\E_f$ is an MDS code over $\gf(q)$ with
parameters $[q+2, 3, q]$ if $f$ is an oval polynomial over $\gf(q)$. This is the classical construction of MDS codes
with oval polynomials. The task of this section is to prove the following theorem.

\begin{theorem}\label{thm-J271}
Let $m \geq 3$, and let $f$ be an oval polynomial over $\gf(q)$. Then the extended code $\bar{\E}_f$ is an NMDS code
over $\gf(q)$ with parameters $[q+3, 3, q]$ and weight enumerator
\begin{eqnarray*}
1+\frac{(q-1)(q+2)}{2} z^q + \frac{(q-1)q(q+2)}{2}z^{q+1} +\frac{(q-1)q}{2} z^{q+2} + \frac{(q-2)(q-1)q}{2} z^{q+3}.
\end{eqnarray*}
\end{theorem}

\begin{proof}
It is well known that $\sum_{x \in \gf(q)} x =0$.
Since $f$ is a permutation on $\gf(q)$, we have
$$
\sum_{x \in \gf(q)} f(x)=0.
$$
Then by definition, the extended code $\bar{\E}_f$ has generator matrix
\begin{eqnarray*}
\bar{B}_f=\left[
\begin{array}{llllllll}
f(0) & f(\alpha^0)  & f(\alpha^1) & \cdots & f(\alpha^{q-2}) & 1 & 0 & 1\\
0 &\alpha^0  & \alpha^1 & \cdots & \alpha^{q-2} & 0 & 1 & 1\\
1      & 1         &   1           & \cdots & 1                    & 0 & 0 & 0
\end{array}
\right].
\end{eqnarray*}

Since $\E_f$ is a $[q+2, 3, q]$ MDS code, the dual code $\E_f^\perp$ is a $[q+2, q-1, 4]$ MDS code over $\gf(q)$.
Therefore, any three columns of $B_f$ are linearly independent over $\gf(q)$.  By definition, $\bar{\E}_f$ has length
$q+3$ and dimension $3$. We need to determine the minimum distance $d(\bar{\E}_f)$. To do this, we first settle
the parameters of the dual code $\bar{\E}_f^\perp$.

The dual code $\bar{\E}_f^\perp$ has length $q+3$ and dimension $q$, as $\bar{\E}_f$ has dimension $3$ and length
$q+3$. Note that the last three columns of $\bar{B}_f$ are linearly dependent over $\gf(q)$. This means that $\bar{\E}_f^\perp$
has codewords of Hamming weight $3$. Thus, the minimum distance $d(\bar{\E}_f^\perp) \leq 3$. Note that no column
of $\bar{B}_f$ is the zero vector. We deduce that $d(\bar{\E}_f^\perp) > 1$. Since $\E_f^\perp$ has minimum distance $4$,
any two columns of $B_f$ are linearly independent over $\gf(q)$. To prove that $d(\bar{\E}_f^\perp) >2$, it suffices to
show that the last column of $\bar{B}_f$, i.e., the vector $(1 1 0)^T$, is linearly independent of any other column of $\bar{B}_f$,
which is obvious. Consequently, we have $d(\bar{\E}_f^\perp) = 3$. Hence, $\bar{\E}_f^\perp$ is an almost MDS code with
parameters $[q+3, q, 3]$.

Since $\E_f$ is a $[q+2, 3, q]$ MDS code, by definition the minimum distance $d(\bar{\E}_f) \geq q$. By the Singleton
bound, $d(\bar{\E}_f) \leq q+1$. If $d(\bar{\E}_f) = q+1$, then $\bar{\E}_f$ would be a $[q+3, 3, q+1]$ MDS code, and
$\bar{\E}_f^\perp$ would be an MDS code with parameters $[q+3, q, 4]$, which is contrary to the proved fact that
$\bar{\E}_f^\perp$ is an almost MDS code with parameters $[q+3, q, 3]$. Thus, $d(\bar{\E}_f) = q$ and
$\bar{\E}_f$ is an almost MDS code with parameters $[q+3, 3, q]$.

Finally, we settle the weight distribution of the code $\bar{\E}_f$. To this end, we first determine the number of
codewords of weight $3$ in the dual code $\bar{\E}_f^\perp$. The discussions above showed that any codeword of weight $3$
in  $\bar{\E}_f^\perp$ must have a nonzero coordinate in the last position. Hence, we count the number of codewords of weight $3$ in $\bar{\E}_f^\perp$ by considering the following four cases.

{\em Case 1:} Consider the following matrix equation
\begin{eqnarray*}
\left[
\begin{array}{lll}
f(x)  & f(y) & 1 \\
x    &   y    &  1 \\
1   &     1   &   0
\end{array}
\right]
\left[
\begin{array}{l}
a \\
b \\
c
\end{array}
\right] = \bzero,
\end{eqnarray*}
where $a, b, c \in \gf(q)^*$, $x, y \in \gf(q)$ and $x \neq y$.  The matrix equation above is the same as the system
of equations:
\begin{eqnarray*}
\left\{
\begin{array}{l}
f(x)+f(y)+c/a =0, \\
x+y+c/a=0, \\
a=b,
\end{array}
\right.
\end{eqnarray*}
which has the same number of solutions  as the following system of equations
\begin{eqnarray}\label{eqn-J2811}
f(x+c/a)+f(x)=c/a, \ \ a=b.
\end{eqnarray}
Since $f$ is an oval polynomial, for any fixed $x \in \gf(q)$ the polynomial $(f(x+z)+f(x))z^{q-2}$ is a permutation on
$\gf(q)$. Hence, there is a unique $z \in \gf(q)$ such that $(f(x+z)+f(x))z^{q-2}=1$. Hence, for each fixed $x \in \gf(q)$
there is a unique $z\in \gf(q)$ such that $f(x+z)+f(x)=z$. It then follows that the number of solutions $(x, a, b, c)$
with $x \in \gf(q)$ and $\{a, b, c\} \subset \gf(q)^*$ of (\ref{eqn-J2811}) is $q(q-1)$.  Therefore, the total number of
codewords of weight 3 whose first two nonzero coordinates are among the first $q$ positions and the last nonzero
coordinate is in the last position is equal to $q(q-1)/2$.

{\em Case 2:} Note that the matrix
\begin{eqnarray*}
\left[
\begin{array}{lll}
f(x) & 1 & 1 \\
x  & 0  & 1 \\
1 & 0 & 0
\end{array}
\right]
\end{eqnarray*}
has rank $3$ for each $x \in \gf(q)$. We deduce that $\bar{\E}_f^\perp$ does not have a codeword of weight 3 whose
first nonzero coordinate is among the first $q$ positions and the remaining two are on the $(q+1)$-th and $(q+3)$-th
positions.

{\em Case 3:} Note that the matrix
\begin{eqnarray*}
\left[
\begin{array}{lll}
f(x) & 0 & 1 \\
x  & 1  & 1 \\
1 & 0 & 0
\end{array}
\right]
\end{eqnarray*}
has rank $3$ for each $x \in \gf(q)$. We deduce that $\bar{\E}_f^\perp$ does not have a codeword of weight 3 whose
first nonzero coordinate is among the first $q$ positions and the remaining two are on the $(q+2)$-th and $(q+3)$-th
positions.

{\em Case 4:} Note that the matrix
\begin{eqnarray*}
\left[
\begin{array}{lll}
1 & 0 & 1 \\
0  & 1  & 1 \\
0 & 0 & 0
\end{array}
\right]
\end{eqnarray*}
has rank $2$. We deduce that $\bar{\E}_f^\perp$ has $q-1$ codewords of weight 3 whose nonzero coordinates are
in the last three positions.

Summarising the conclusions in Cases 1--4, we know that the total number of codewords of weight 3 in
$\bar{\E}_f^\perp$ is $(q-1)(q+2)/2$. By Theorem \ref{thm-121FW}, the number of codewords of weight $q$ in
$\bar{\E}_f$ is $(q-1)(q+2)/2$. The desired weight enumerator of $\bar{\E}_f$ then follows from Theorem \ref{thm-DLwtd}.
\end{proof}

\begin{example}
Let $m=3$. Then the code $\bar{\E}_{x^6}$ over $\gf(2^3)$ has parameters $[11, 3, 8]$ and weight enumerator
$
1 + 35z^8 + 280 z^9  +  28z^{10} + 168z^{11}.
$
\end{example}

Notice that the construction of NMDS codes in Theorem \ref{thm-J271} works for every oval polynomial over $\gf(q)$,
and is thus general.
With the known nine infinite families of oval polynomials documented in Theorem \ref{thm-knownopolys},
nine infinite families of $[q+3, 3, q]$ NMDS codes over $\gf(q)$ are obtained via Theorem \ref{thm-J271}.
For any arc $\cA$ of $\PG(2, q)$, it is well known that $|\cA| \leq q+2$. Hence, there is no MDS code over
$\gf(q)$ with parameters $[q+3, 3, q+1]$, and these nine infinite families of  $[q+3, 3, q]$ NMDS codes over $\gf(q)$ are thus
distance-optimal.

\section{Near MDS codes with parameters $[q+1, 3, q-2]$ from oval polynomials}\label{sec-extMDSc}

Let $f$ be a polynomial over $\gf(q)$ with $f(0)=0$ and $f(1)=1$. Let $\alpha$ be a generator of $\gf(q)^*$.
Define
\begin{eqnarray}
G_f=\left[
\begin{array}{llllll}
f(\alpha^0)  & f(\alpha^1) & \cdots & f(\alpha^{q-2}) & 0 & 1 \\
\alpha^0  & \alpha^1 & \cdots & \alpha^{q-2} & 1 & 0 \\
1               &   1           & \cdots & 1                    & 1 & 1
\end{array}
\right].
\end{eqnarray}
By definition, $G_f$ is a $3$ by $q+1$ matrix over $\gf(q)$. Let $\C_f$ denote the linear code over $\gf(q)$ with
generator matrix $G_f$.

\begin{theorem}\label{thm-nmdscodej231}
Let $m \geq 3$ be odd and let $f(x)$ be an oval polynomial over $\gf(q)$ with coefficients in $\gf(2)$. Then $\C_f$ is a $[q+1, 3, q-2]$ NMDS code
over $\gf(q)$ with weight enumerator
\begin{eqnarray*}
A(z)=1 + (q-1)(q-2)z^{q-2} + \frac{(q-1)(q^2-5q+12)}{2} z^{q-1} + \\
(q-1)(4q-5) z^{q} + \frac{(q-1)(q^2-3q+4)}{2} z^{q+1}.
\end{eqnarray*}
\end{theorem}

\begin{proof}
We first prove that the dimension $\dim(\C_f)$ of $\C_f$ is $3$.
Let $\bg_1$, $\bg_2$ and $\bg_3$ denote the first, second and third rows of $G_f$, respectively.
Assume that $a \bg_1 + b \bg_2 + c \bg_3 =0$ for three elements $a, b$ and $c$ in $\gf(q)$,
where at least one of the elements in $\{a,b,c\}$ is nonzero.
By the definition of the last two columns of $G_f$, any two rows of $G_f$ are linearly independent
over $\gf(q)$. Consequently, $\dim(\C_f) \geq 2$ and $abc \ne 0$.  It then follows from  $a \bg_1 + b \bg_2 + c \bg_3 =0$
that
\begin{eqnarray}
\left\{
\begin{array}{l}
a=b=c \ne 0, \\
f(x)+x+1=0 \mbox{ for all } x \in \gf(q)^*.
\end{array}
\right.
\end{eqnarray}
We have then $f(1)=0$, which is contrary to our assumption that $f(1)=1$. Therefore, $\dim(\C_f)=3$. Notice
that only the conditions that $f(0)=0$ and $f(1)=1$ guarantee that the dimension of the code $\C_f$ is $3$.

\subsubsection*{We now prove that  $\C_f^\perp$ has parameters $[q+1, q-2, 3]$.}

Clearly $\dim(\C_f^\perp)=q+1-\dim(\C_f)=q-2$. Since no column of $G_f$ is the zero vector, the minimum distance
$d(\C_f^\perp) \geq 2$.
It is straightforward to prove that any two columns of $G_f$ are linearly independent over $\gf(q)$. Hence,
$d(\C_f^\perp) > 2$. We now prove that $d(\C_f^\perp) = 3$, and compute the total number of codewords of weight
$3$ in $\C_f^\perp$. We need to consider several cases below.

{\it Case 1.1:} Let $x \in \gf(q)^*$. Consider the following matrix
\begin{eqnarray*}
M_{1,1}=\left[
\begin{array}{lll}
f(x) & 0 & 1 \\
x    &  1 & 0 \\
1   &   1 &  1
\end{array}
\right],
\end{eqnarray*}
which is a submatrix of the generator matrix $G_f$. By Theorem \ref{thm-J2220}, $f(x)+x+1 \neq 0$. Therefore,
$M_{1,1}$ has rank $3$. Consequently, $\C_f^\perp$ does not have a codeword of weight $3$ whose last two coordinates
are nonzero.

{\it Case 1.2:} Let $x, y, z$ be three pairwise distinct elements in $\gf(q)^*$. Consider the following matrix
\begin{eqnarray*}
M_{1,2,1}=\left[
\begin{array}{lll}
f(x) & f(y) & f(z) \\
x    &  y & z \\
1   &   1 &  1
\end{array}
\right],
\end{eqnarray*}
which is a submatrix of the generator matrix $G_f$. Note that $M_{1,2,1}$ has the same rank as the matrix
\begin{eqnarray*}
M_{1,2,2}=\left[
\begin{array}{lll}
f(x)+f(z) & f(y)+f(z) & f(z) \\
x+z    &  y+z & z \\
0  &   0 &  1
\end{array}
\right].
\end{eqnarray*}
We have
$$
|M_{1,2,2}|=(f(x)+f(z))(y+z)+(f(y)+f(z))(x+z).
$$
By Theorem \ref{thm-J22220}, $|M_{1,2,2}| \neq 0$. Therefore,
$M_{1,2,1}$ has rank $3$.
Consequently, $\C_f^\perp$ does not have a codeword of weight $3$ whose nonzero coordinates
are in the first $q-1$ positions.

{\it Case 1.3:} Let $x, y,$ be two distinct elements in $\gf(q)^*$. Consider the following matrix
\begin{eqnarray*}
M_{1,3,1}=\left[
\begin{array}{lll}
f(x) & f(y) & 1 \\
x    &  y & 0 \\
1   &   1 &  1
\end{array}
\right],
\end{eqnarray*}
which is a submatrix of the generator matrix $G_f$. Note that $M_{1,3,1}$ has the same rank as the matrix
\begin{eqnarray*}
M_{1,3,2}=\left[
\begin{array}{lll}
f(x)+1 & f(y)+1 & 1 \\
x    &  y & 0 \\
0  &   0 &  1
\end{array}
\right].
\end{eqnarray*}
We have
$$
|M_{1,3,2}|=(f(x)+1)y+(f(y)+1)x.
$$
If one of $x$ and $y$ is $1$, then $|M_{1,3,2}| \neq 0$ as $x \neq y$. We now calculate the pairs of distinct
$x$ and $y$ in $\gf(q) \setminus \{0, 1\}$ such that $|M_{1,3,2}| =0$.

For each $x \in \gf(q) \setminus \{0, 1\}$, let $a=(f(x)+1)/x$. Then $a \neq 0$. By Theorem \ref{thm-J2220},
$a \neq 1$.
By Theorem \ref{thm-opoly2to1},
$f(z)+az$ is $2$-to-$1$. Thus, there is another unique element $y \in \gf(q) \setminus \{0, 1\}$ such that
$$
f(y)+ay=1=f(x)+ax.
$$
Thus for this pair of distinct $x$ and $y$, we have $|M_{1,3,2}| =0$. Hence, the total number of distinct pairs
$x$ and $y$ in $ \gf(q) \setminus \{0, 1\}$ such that $|M_{1,3,2}| =0$ is equal to $(q-2)/2$.
Consequently, in $\C_f^\perp$  the total number of codewords of weight $3$ whose two nonzero coordinates
are in the first $q-1$ positions and the other nonzero coordinate is in the $(q+1)$-th positionn is equal to
$(q-2)(q-1)/2$.

{\it Case 1.4:} Let $x, y,$ be two distinct elements in $\gf(q)^*$. Consider the following matrix
\begin{eqnarray*}
M_{1,4,1}=\left[
\begin{array}{lll}
f(x) & f(y) & 0 \\
x    &  y & 1 \\
1   &   1 &  1
\end{array}
\right],
\end{eqnarray*}
which is a submatrix of the generator matrix $G_f$. Note that $M_{1,4,1}$ has the same rank as the matrix
\begin{eqnarray*}
M_{1,4,2}=\left[
\begin{array}{lll}
f(x) & f(y) & 0 \\
x +1   &  y+1 & 1 \\
0  &   0 &  1
\end{array}
\right].
\end{eqnarray*}
We have
$$
|M_{1,4,2}|=f(x)(y+1)+f(y)(x+1).
$$

Choose any $y \in \gf(q) \setminus \{0,1\}$. Define $a=f(y)/(y+1)$. Then $a \neq 0$. By Theorem \ref{thm-J2220},
$a \neq 1$. Note that
$$
f(y)+ay=a.
$$
By Theorem \ref{thm-opoly2to1}, $f(z)+az$ is 2-to-1. Hence, there is an element $x \in \gf(q)^*$ such that $x \neq y$
and
$$
f(x)+ax=a.
$$
For this pair $(x, y)$,
$$
|M_{1,4,2}|=f(x)(y+1)+f(y)(x+1) = 0.
$$

Conversely, let $x$ and $y$ be two distinct elements in $\gf(q)^*$ such that
$$
|M_{1,4,2}|=f(x)(y+1)+f(y)(x+1) = 0.
$$
Then
$$
\frac{f(x)}{x+1} = \frac{f(y)}{y+1} = a
$$
for some $a \in \gf(q)$. Since $x \neq y$ and $f$ is bijective, $a \neq 0$. Thus, the total number of distinct $x$ and
$y$ in $\gf(q)^*$ such that $|M_{1,4,2}|=0$ is equal to $(q-2)/2$.
Consequently, in $\C_f^\perp$  the total number of codewords of weight $3$ whose two nonzero coordinates
are in the first $q-1$ positions and the other nonzero coordinate is in the $q$-th positionn is equal to
$(q-2)(q-1)/2$.

Summarizing the discussions in Cases 1.1, 1.2, 1.3 and 1.4, we deduce that the total number of codewords of weight
$3$ in $\C_f^\perp$ is $(q-1)(q-2)$ and the minimum distance $d(\C_f^\perp)=3$. Thus, $\C_f^\perp$ has parameters
$[q+1, q-2, 3]$ and is an almost MDS code.

\subsubsection*{We then prove that  the minimum distance $d(\C_f) = q-2$.}

On the contrary, suppose that  $d(\C_f) \leq q-3=q+1-4$. Let $\bc=a \bg_1 + b \bg_2 + c \bg_3$ be a minimum
weight codeword in $\C_f$. Then, at least four coordinates in $\bc$ are zero. We now consider the following two
cases.

{\it Case 2.1: Suppose that the last two coordinates in $\bc$ are zero.} Then there exist two distinct elements
$x$ and $y$ in $\gf(q)$ such that
\begin{eqnarray}\label{eqn-j221}
\left\{
\begin{array}{r}
af(x)+bx+c = 0, \\
af(y)+by+c = 0, \\
b+c=0, \\
a+c=0,  \\
\end{array}
\right.
\end{eqnarray}
where $a, b, c$ are the constants for defining the minimum weight codeword $\bc=a\bg_1 + b\bg_2 + c \bg_3$.
It follows from (\ref{eqn-j221}) that
$$
f(x)+x+1=0 \mbox{ and } f(y)+y+1=0.
$$
This is contrary to Theorem \ref{thm-J2220}.

{\it Case 2.2: Suppose that at most one of the last two coordinates in $\bc$ is zero.} In this case, there are three pairwise
distinct elements $x, y, z$ in $\gf(q)^*$ such that
\begin{eqnarray}
\left[
\begin{array}{ccc}
f(x) & x & 1 \\
f(y) & y & 1 \\
f(z) & z & 1 \\
\end{array}
\right]
\left[
\begin{array}{c}
a \\
b \\
c
\end{array}
\right]
=0.
\end{eqnarray}
Clearly, the rank of the matrix
\begin{eqnarray*}
M_1 = \left[
\begin{array}{ccc}
f(x) & x & 1 \\
f(y) & y & 1 \\
f(z) & z & 1 \\
\end{array}
\right]
\end{eqnarray*}
is the same as the rank of the matrix
\begin{eqnarray*}
M_2 = \left[
\begin{array}{ccc}
f(x)+f(z) & x+z & 0 \\
f(y)+f(z) & y+z &  0\\
f(z) & z & 1 \\
\end{array}
\right].
\end{eqnarray*}
It follows from Theorem \ref{thm-J22220} that the determinant
$$
|M_2|=(f(x)+f(z))(y+z)+(f(y)+f(z))(x+z) \neq 0.
$$
Since $M_1$ has full rank, $a=b=c=0$ and $\bc=0$. This is contrary to the fact that $\bc$ is a minimum weight codeword
in $\C$.

Summarizing the discussions in Cases 2.1 and 2.2 proved that $d(\C_f) \geq q-2$. By the Singleton bound,  $d(\C_f) \leq q-1$.
If $d(\C_f) = q-1$, then $\C_f$ would be an MDS code with parameters $[q+1, 3, q-1]$ and $\C_f^\perp$ would be an MDS
code with parameters $[q+1, q-2, 4]$, which is contrary to the proved fact that $d(\C_f^\perp)=3$. We then arrived at the
conclusion that $d(\C_f) = q-2$. Consequently, $\C_f$ is an almost MDS code with parameters $[q+1, 3, q-2]$. By definition,
$\C_f$ is an NMDS code.

It then follows from Theorem \ref{thm-121FW} that the total number $A_{q-2}$ of minimum weight codewords in $\C_f$
is equal to the total number of codewords of weight $3$ in $\C_f^\perp$, and is $(q-1)(q-2)$. The desired conclusion
on the weight enumerator of $\C_f$ then follows from  Theorem \ref{thm-DLwtd}. This completes the proof of this theorem.
\end{proof}

\begin{example}
Let $m=3$ and $f(x)=x^6$. Then the code  $\C_f$ has parameters $[9, 3, 6]$ and weight enumerator
$$
1+42z^6 + 126z^7 + 189z^8  +  154z^{9}.
$$
\end{example}

With the first seven families of oval polynomials documented in Theorem \ref{thm-knownopolys},  we have constructed seven infinite
families of near MDS codes over $\gf(q)$ with parameters $[q+1, 3, q-2]$ via Theorem \ref{thm-nmdscodej231}.
Note that this construction may not work for the Subiaco and Adelaide oval polynomials in general.

\section{Near MDS codes with parameters $[q+2, 3, q-1]$ from oval polynomials}\label{sec-J272}

Let $f$ be a polynomial over $\gf(q)$ with $f(0)=0$ and $f(1)=1$. Let $\alpha$ be a generator of $\gf(q)^*$.
Define
\begin{eqnarray}
\bar{G}_f=\left[
\begin{array}{lllllll}
f(0) &f(\alpha^0)  & f(\alpha^1) & \cdots & f(\alpha^{q-2}) & 0 & 1 \\
0 &\alpha^0  & \alpha^1 & \cdots & \alpha^{q-2} & 1 & 0 \\
1 &1               &   1           & \cdots & 1                    & 1 & 1
\end{array}
\right].
\end{eqnarray}
By definition, $\bar{G}_f$ is a $3$ by $q+2$ matrix over $\gf(q)$. Let $\bar{\C}_f$ denote the linear code over $\gf(q)$ with
generator matrix $\bar{G}_f$.

\begin{theorem}\label{thm-nmdscodej232}
Let $m \geq 3$ be odd and let $f(x)$ be an oval polynomial over $\gf(q)$ with coefficients in $\gf(2)$. Then $\bar{\C}_f$ is a $[q+2, 3, q-1]$ NMDS code
over $\gf(q)$ with weight enumerator
\begin{eqnarray*}
\bar{A}(z)=1 + (q-1)(q-2)z^{q-1} + \frac{(q-1)(q^2-3q+14)}{2} z^{q} + \\
3(q-1)(q-2) z^{q+1} + \frac{(q-1)(q^2-3q+4)}{2} z^{q+2}.
\end{eqnarray*}
\end{theorem}

A slight modification of the proof of Theorem \ref{thm-nmdscodej231} gives a proof of Theorem \ref{thm-nmdscodej232}.
The details of the proof are omitted here.

\begin{example}
Let $m=3$ and $f(x)=x^6$. Then the code  $\bar{\C}_f$ has parameters $[10, 3, 7]$ and weight enumerator
$$
1+42z^7 + 189z^8 + 126z^9  +  154z^{10}.
$$
\end{example}

With the first seven families of oval polynomials documented in Theorem \ref{thm-knownopolys},  we have constructed seven infinite
families of near MDS codes over $\gf(q)$ with parameters $[q+2, 3, q-1]$ via Theorem
\ref{thm-nmdscodej232}. Note that the construction of this section may not work for
the Subiaco and Adelaide oval polynomials in general.

\section{Summary and concluding remarks}

Let $r$ be a prime power. For any $[n, k, d]$ linear code $\C$ over $\gf(r)$, the Griesmer bound says that
$$
n \geq \sum_{i=0}^k \left\lceil \frac{d}{r^i} \right\rceil.
$$
An $[n, k, d]$ linear code $\C$ over $\gf(r)$ is said to be \emph{almost optimal} with respect to the  Griesmer bound if
$$
n-1 = \sum_{i=0}^k \left\lceil \frac{d}{r^i} \right\rceil.
$$

The contributions of this paper are the following:
\begin{itemize}
\item Nine infinite families of near MDS codes over $\gf(2^m)$ with parameters $[2^m+3, 3, 2^m]$ via
          Theorems \ref{thm-J271} and \ref{thm-knownopolys}, which are distance-optimal.
          The construction of the MDS codes over $\gf(q)$ with parameters $[q+2, 3, q]$ in Section \ref{sec-J271}
is classical and thus not new.
Our contribution in Section \ref{sec-J271} is to prove that their extended codes are near MDS and to settle the weight
distribution of the near MDS codes.
\item Seven infinite families of near MDS codes over $\gf(2^m)$ with parameters $[2^m+1, 3, 2^m-2]$ via
         Theorems \ref{thm-nmdscodej231} and \ref{thm-knownopolys}, which are almost optimal with respect to
         both the Singleton and Griesmer bounds.
\item  Seven infinite families of near MDS codes over $\gf(2^m)$ with parameters $[2^m+2, 3, 2^m-1]$ via
       Theorems \ref{thm-nmdscodej232} and \ref{thm-knownopolys}, which are almost optimal with respect to
         both the Singleton and Griesmer bounds.
\end{itemize}

We remark that our constructions of near MDS codes with oval polynomials in Sections  \ref{sec-extMDSc}
and \ref{sec-J272} are similar to the classical construction of NMDS codes with oval polynomials in Section \ref{sec-J271}.
The constructions of near MDS codes presented in Sections  \ref{sec-extMDSc} and \ref{sec-J272} work
for oval polynomials over $\gf(2^m)$ with coefficients in $\gf(2)$ and odd $m$ only, while the classical construction of near MDS codes  in Section \ref{sec-J271} works
for all oval polynomials over $\gf(2^m)$ for both odd and even $m$. This shows a big difference between the
constructions in Sections  \ref{sec-extMDSc} and \ref{sec-J272} and the classical construction in Section \ref{sec-J271}. Of course, the three constructions produce NMDS codes with different parameters.

It would be a nice problem to investigate applications of the near MDS codes of this paper in cryptography following the ideas in \cite{LW17} and \cite{Zhou09}.


\end{document}